\documentclass{amsart}

\usepackage{amsmath, amsfonts, xspace, enumerate}
\usepackage{tikz, fancyhdr}
\usetikzlibrary{decorations.pathreplacing}

\newcommand{\nceil}[1]{\lceil #1\rceil}

\newcommand{\nsset}[1]{\{#1\}}
\newcommand{\nabs}[1]{\lvert #1\rvert}
\newcommand{\nfloor}[1]{\lfloor #1\rfloor}

\newcommand{\Z}{\ensuremath{\mathbb{Z}}\xspace}
\newcommand{\calP}{\ensuremath{\mathcal{P}}\xspace}
\newcommand{\calS}{\ensuremath{\mathcal{S}}\xspace}
\newcommand{\calL}{\ensuremath{\mathcal{L}}\xspace}
\newcommand{\word}[1]{\ensuremath{\mathtt{#1}}\xspace}

\theoremstyle{plain}
\newtheorem{thm}{Theorem}
\newtheorem{lem}[thm]{Lemma}
\theoremstyle{definition}
\newtheorem{defn}[thm]{Definition}

\begin{document}
\usetikzlibrary{decorations.pathreplacing}

\title{A new geometric approach to Sturmian words}

\author{Kaisa Matom\"aki}
\author{Kalle Saari}

\address{Department of Mathematics \\ University of Turku \\ 20014 Turku \\ Finland}
\email{ksmato@utu.fi}
\email{kasaar@utu.fi}
\thanks{The research of the first author was supported by grant no.\ 137883 from the Academy of Finland and the research of the second author by grant no.\ 134190 from the Academy of Finland.}

\begin{abstract}
We introduce a new geometric approach to Sturmian words by means of a mapping that associates certain lines in the $n\times n$-grid and sets of finite Sturmian words of length~$n$.
Using this mapping, we give new proofs of the formulas enumerating the finite Sturmian words and the palindromic finite Sturmian words of a given length.
We also give a new proof for the well-known result that a factor of a Sturmian word has precisely two return words.
\end{abstract}
\maketitle

\section{Introduction}
\label{se:intro}
An infinite binary word is \emph{Sturmian} if it has exactly $n+1$ factors of length $n$ for all integers $n\geq 0$. Factors of Sturmian words are called \emph{finite Sturmian words}. For definitions and notation, see Lothaire~\cite{Lothaire2002}. 

The purpose of this note is to introduce a new geometric approach to Sturmian words by means of a mapping that associates certain lines in the $n\times n$-grid and sets of finite Sturmian words of length~$n$. We demonstrate the power of the approach by giving new proofs for enumeration formulas for finite Sturmian words and palindromes as well as for the fact that a factor of Sturmian word has precisely two returns. We believe that this new approach can be used to study also other properties of Sturmian words. 

Let us describe the stated applications more precisely before going to the details of the geometric approach in Section~\ref{mapping}. Let $A_{n}$ denote the set of finite Sturmian words of length $n\geq 0$. The formula
\begin{equation}\label{011120101513}
\nabs{A_{n}} = 1 + \sum_{k=1}^{n}(n+1 - k)\varphi(k),
\end{equation}
where  $\varphi$ is Euler's totient function, 
seems to have been proven first by Lipatov~\cite{Lipatov1982}. Subsequently, proofs have been given by Mignosi~\cite{Mignosi1991}  by means of Farey sequences,
by Berstel and Pocchiola~\cite{BerPoc1993} with geometric arguments, and de Luca and Mignosi~\cite{deLMig1994} with combinatorial arguments. 
In Section~\ref{finite} we give an alternate geometric proof of~\eqref{011120101513} based on the setting introduced in Section \ref{mapping}. Our approach partially parallels to that of Berstel and Pocchiola, but is somewhat more elementary;
in particular, it avoids the use of Euler's formula that relates the number of edges, faces, and vertices of a planar graph.

It was shown by de Luca and De Luca~\cite{deLDeL2006}, using combinatorial arguments, that the number of palindromic finite Sturmian words  of length~$n$ is
\begin{equation}\label{10062011}
1+ \sum_{k=0}^{\nceil{n/2}-1} \varphi(n-2k).
\end{equation} 
Applying the new geometric approach, we give a new proof of~\eqref{10062011} in Section~\ref{palindrome}.

Let $\omega$ be an infinite word and $u$ its factor. A word $v$ is called a \emph{return word of $u$ in $\omega$} if $vu$ occurs in $\omega$ and it has precisely two occurrences of $u$ in it: one as a prefix and the other as a suffix. It is well-known that each factor of an infinite Sturmian word has exactly two return words, see \cite{JusVui2000,Vuillon2001}. In Section~\ref{return}, we give a new proof for this result, as well.

\section{A mapping relating lines and Sturmian words} \label{mapping}

We start by relating Sturmian words to lines: First we will describe a well-known geometric interpretation of Sturmian words and then utilize it to find a mapping which relates every finite Sturmian word of length $n$ with a line which has at least two integer points in $n \times n$-grid. 

An infinite word $\omega = a_{0} a_{1} a_{2} \cdots a_{k} \cdots $ with $a_{k} \in \{ \mathtt{0}, \mathtt{1} \}$ is Sturmian if and only if there exist an irrational number $\alpha \in (0,1)$ and a real number $\rho$ such that 
\[
a_{k} = \nfloor{(k+1)\alpha + \rho} - \nfloor{k\alpha + \rho}
\]
for all $k\geq 0$ (see \cite[Ch.~2]{Lothaire2002}). 
Since every finite Sturmian word is a prefix of an infinite Sturmian word, 
this implies that a finite word $w=a_{0} a_{1} a_{2} \cdots a_{n-1}$ with $a_{k} \in \{ \mathtt{0}, \mathtt{1}\}$ is a finite Sturmian word if and only if there exist real numbers $\alpha\in (0,1)$ and  $\rho\in (0,1)$ such that 
\begin{equation}
\label{eq:akdef}
a_{k} = \nfloor{(k+1)\alpha + \rho} - \nfloor{k\alpha + \rho} \qquad (0 \leq k \leq n-1).
\end{equation}
Note that, while it is clear that one can restrict to $\rho \in [0, 1)$, in case of finite Sturmian words one can indeed restrict to $\rho \in (0, 1)$: If a word $w$ is obtained from \eqref{eq:akdef} with $\rho = 0$ and some $\alpha \in (0,1)$, one can increase $\rho$ slightly to make it positive; If this change is small enough, the letters $a_{k}$ are unaffected for $0 \leq k \leq n-1$. Furthermore, unlike in the case of an infinite Sturmian word, here $\alpha$ may also be rational because in that case one can increase $\alpha$ slightly to make it irrational: Again, if this change is small enough, the letters $a_{k}$ are unaffected for $0 \leq k \leq n-1$, and so the obtained $w$ is indeed a finite Sturmian word. 

The letter $a_{k}$ can be determined geometrically as follows. Consider the line~$\ell$ given by the equation $y = \alpha x + \rho$ drawn in the  integer grid $[0,n] \times [0,n]$ (or simply $n\times n$-grid). 
The letter $a_{k}$ equals $\mathtt{1}$ precisely when $\ell$ intersects  a horizontal bar (or a grid point) at some $x \in (k, k+1]$; see Figure~\ref{041120101629} for an example.  
Thus any finite Sturmian word $a_{0} a_{1} a_{2} \cdots a_{n-1}$ may be identified with a broken line starting from $(0,0)$ in the $n \times n$-grid
such that $a_{k}$ corresponds to a diagonal line if $a_{k} = \word{1}$ and a horizontal line if $a_{k} = \word{0}$;  again, consult Figure~\ref{041120101629} for an example.

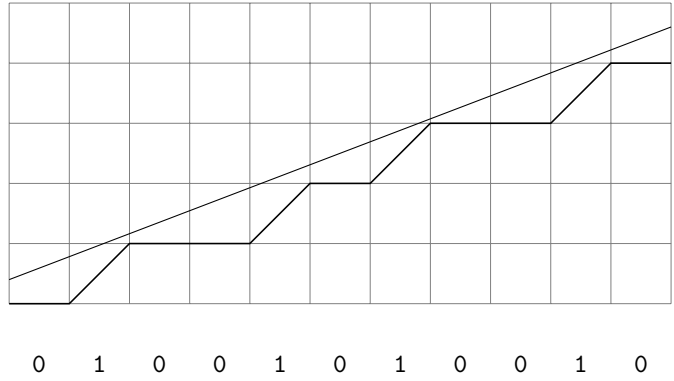
\begin{figure}[h]
 \centering
 \begin{tikzpicture}[scale = 0.8]
 \draw[step=1cm,gray,very thin] (0,0) grid (11,5); 
 \draw (0,0.4) -- (11, 4.602);
 \draw (0.5, -1) node {\word{0}};
 \draw (1.5, -1) node {\word{1}};
 \draw (2.5, -1) node {\word{0}};
 \draw (3.5, -1) node {\word{0}};
 \draw (4.5, -1) node {\word{1}};
 \draw (5.5, -1) node {\word{0}};
 \draw (6.5, -1) node {\word{1}};
 \draw (7.5, -1) node {\word{0}};
 \draw (8.5, -1) node {\word{0}}; 
 \draw (9.5, -1) node {\word{1}}; 
 \draw (10.5, -1) node {\word{0}};
 \draw[semithick] (0, 0.0) -- (1, 0) -- (2, 1) -- (3, 1) -- (4, 1) -- (5, 2) -- (6, 2) -- (7, 3) -- (8, 3) -- (9, 3) -- (10, 4) -- (11, 4);
 \end{tikzpicture}
 \caption{Identifying a finite Sturmian word with a broken line. Here $\alpha = (3- \sqrt{5})/2$, $\rho = 2/5$, and $n=11$.}\label{041120101629}
 \end{figure}

In what follows, we identify a finite Sturmian word $w=a_{0} a_{1} a_{2} \cdots a_{n-1}$ with a broken line as explained above.
Hence we may say that \emph{$w$ goes through the point $p = (i,j)$}, where $i$ and $j$ are integers with $0\leq i,j \leq n$, if the number of
$\word{1}$s in the first $i$ letters of $w$ is $j$. For example, the word in Figure~\ref{041120101629} goes through the point $(7,3)$. 
We will often use $i$ and $j$ to denote the coordinates of a point~$p$ and thus write $p_1 = (i_1, j_1)$, for example.

Let $\calS_n$ be the set of all lines in the $n\times n$-grid with equations of the form
\[
y = \alpha x +\rho \quad \text{with} \quad \alpha \in (0,1) \quad\text{and} \quad \rho \in (0,1).
\]
It follows from the identification of finite Sturmian words with broken lines that, on the one hand, each line in $\calS_n$ defines a finite Sturmian word of length $n$ (even if $\alpha$ is rational) and, on the 
other hand, each word is determined by some lines in $\calS_n$. 

Next we will define another set of lines in the $n \times n$-grid. For any line $\ell$ and real number $k$, let $\Z_k(\ell)$ be the number of integer points
$(i,j) \in \Z \times \Z$ on the line with $0 \leq i \leq k$. Now we let $\calL_n$ be the set of
all lines 
\[
y = \alpha x +\rho \quad \text{with} \quad \alpha \in [0,1], \quad \rho \in [0,1], \quad \text{and} \quad \Z_n(\ell) \geq 2.
\]

Informally speaking, we will assosiate to each $\ell \in \calL_n$ the set of those Sturmian words that are defined by lines $\ell' \in \calS_n$ that go just above more than half of the integer points of $\ell$ and just below the rest.  By ``just above a point'' (or just below) we mean that $\ell'$ goes so little above (or below) the point that there are no integer points between $\ell$ and $\ell'$ anywhere in the grid. For a formal definition recall that $A_n$ denotes the set of finite Sturmian words of length $n$ and write $2^{A_n}$ for the power set of $A_n$.
\begin{defn}
Let $m \colon \calL_n \to  2^{A_n}$ be the mapping for which one has $w \in m(\ell)$ (with $\ell \in \calL_n$ and $w \in A_n$) if and only if there exists a line $\ell' \in \calS_n$ such that
\begin{enumerate}[(i)] 
\item The line $\ell'$ defines the word $w$;
\item There are no grid points between the lines $\ell$ and $\ell'$;
\item The line $\ell'$ goes above two integer points $p_1, p_2 \in \ell$ with $i_1 \leq n/2 < i_2$. 
\end{enumerate}
\end{defn}


For example, if $\ell \in \calL_{10}$ is given by  $y = \frac{1}{2} x + \frac{1}{2}$, then
\[
m(\ell) = \nsset{\word{1010101010}, \word{0110101010}, \word{0101101010}, \word{1010101001}}.
\]
The word $\word{1010101001}$ is given by the line $l' \in \calS_{10}$ depicted in Figure~\ref{091120101629}.
In this case we may choose $p_{1} = (3,2) $ and $p_{2} = (7,4)$.
 \begin{figure}[h]
 \centering
 \begin{tikzpicture} [scale=0.8]
 \draw[step=1cm,gray,very thin] (0,0) grid (10, 6); 
\draw (0,  0.5) -- (10, 5.5);
\draw (0,  0.95) -- (10, 5.4);
\draw (0, 0.5) node[anchor=east]{$\ell$};
\draw (0, 0.95) node[anchor=east]{$\ell'$};
 \filldraw (3, 2) circle (2pt); 
 \filldraw (7, 4) circle (2pt); 
 \draw (0.5,-0.5) node {\word{1}};
 \draw (1.5,-0.5) node {\word{0}};
 \draw (2.5,-0.5) node {\word{1}};
 \draw (3.5,-0.5) node {\word{0}};
 \draw (4.5,-0.5) node {\word{1}};
 \draw (5.5,-0.5) node {\word{0}};
 \draw (6.5,-0.5) node {\word{1}};
 \draw (7.5,-0.5) node {\word{0}};
 \draw (8.5,-0.5) node {\word{0}};
 \draw (9.5,-0.5) node {\word{1}};
  \draw[semithick] (0.0, 0.0) -- (1, 1) -- (2, 1) -- (3,2) -- (4,2) -- (5,3) -- (6,3) -- (7,4) -- (8,4) -- (9,4) -- (10, 5);
 \end{tikzpicture}
 \caption{Line $l \in \calL_{10}$ given by $\frac{1}{2}x + \frac{1}{2}$ and a line $l' \in \calS_{10}$ that goes just above the points $p_{1} = (3,2) $ and $p_{2} = (7,4)$ giving the word $\word{1010101001}$.}\label{091120101629}
 \end{figure}
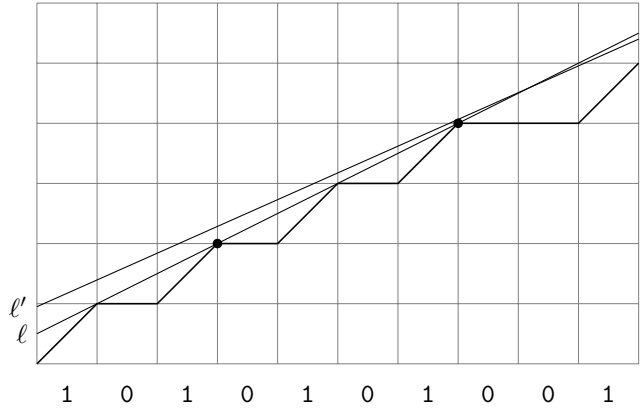

The next lemma is crucial to our work as it shows that the mapping $m$ gives a very nice correspondence between $\calL_n$ and $A_n$.
\begin{lem}
\label{le:Anml}
For all $n \geq 1$, the sets $m(\ell)$ with $\ell \in \calL_n$ form a partition of $A_n$. In particular, we have
\begin{equation}
\label{eq:Ansumml}
\nabs{A_n} = \sum_{\ell \in \calL_n} \nabs{m(\ell)}.
\end{equation}
\end{lem}

\begin{proof}
First, we show that $\bigcup_{\ell \in \calL_{n}} m(\ell) = A_{n}$.
Consider a Sturmian word $w\in A_{n}$ and let $\alpha x + \rho\in \calS_n$ be a line defining it. Without changing the word~$w$, we can
decrease $\rho$ until the line touches at least one integer point
$(i,j)$. If $i \leq n/2$, we let this point be $p_1$ and we rotate the line
clockwise around $(i,j)$ until we reach a line with at least two
integer points, and hence a line in $\calL_n$. The line will meet from
the above an integer point $(i', j')$ with $i' > n/2$, which we call $p_2$. If $i > n/2$, then we take $(i, j) = p_2$, rotate
the line counter-clockwise and proceed similarly. In either case we reach a line in $\calL_n$ which maps to the word~$w$. 

As an example, consider the situation in Figure \ref{091120101629}. When $l'$ is moved down-wards, the first integer point it meets is $(7, 4)$. Then, when the line is rotated counter-clockwise around $(7,4)$, the point $(3,2)$ is among the next integer points the line meets. 

Second, we show that $m(\ell) \cap m(\ell') = \emptyset$ for all distinct $\ell, \ell'\in \calL_{n}$.
Suppose contrary that $w \in m(\ell) \cap m(\ell')$. Then there are
points $p_1, p_2 \in \ell$ and $p_1', p_2' \in \ell'$ such that $i_1
\leq n/2 < i_2$ and $i_1' \leq n/2 < i'_2$, and the broken line of $w$ goes
through points $p_1, p_2, p_1'$, and $p_2'$. This implies that $p_1$ and $p_2$
must be below $\ell'$ and $p_1'$ and $p_2'$ must be below $\ell$. This
is possible only if the lines intersect to the right from $p_1$ and
$p_2$ and to the left from $p_1'$ and $p_2'$ or vice versa. However,
this is not possible since $i_1 < i_2'$ and $i_1' < i_2$, a contradiction.

Again, as an example, consider Figure \ref{091120101629}. It is clear that there cannot be a line different from $\ell$ that goes above both $(3,2)$ and $(7, 4)$ and through integer points of the word on both sides of $x = 10/2 = 5$.  
\end{proof}

Next we show that in situations where one is allowed to swap to considering extensions of original words, we can restrict to the technically simpler case where the word goes through all grid points of the corresponding line.
\begin{lem}
\label{le:redallints}
Let $w$ be a finite Sturmian word. Then there exists a finite Sturmian word $w'$ such that $w$ is a prefix of $w'$ and $w'$ goes through all grid points of the line $\ell' \in \calL_{|w'|}$ for which $w' \in m(\ell')$.
\end{lem}
\begin{proof}
Let $w$ be of length $n$ and defined by a line $y = \alpha x + \rho$, where $\alpha$ is irrational and $\alpha, \rho \in (0, 1)$. By increasing $\rho$ slightly if needed we can assume that the line has no integer points in the $n \times n$-grid. Let $\omega$ be the corresponding infinite word.

Let $n_0$ be the smallest integer $> n$ such that
\[
\alpha n_0 + \rho - \nfloor{\alpha n_0+\rho} < \min_{0 \leq m \leq n} (\alpha m + \rho - \nfloor{\alpha m+\rho}).
\]
Such $n_0$ exists since the right hand side is positive and the irrationality of $\alpha$ implies that $\{\alpha k \pmod{1} \colon {k \in \mathbb{N}}\}$ is dense in $\mathbb{R} / \mathbb{Z}$ (see \cite[Theorem 439]{HarWri1960}). Notice that now
\[
\alpha n_0 + \rho - \nfloor{\alpha n_0+\rho} = \min_{0 \leq m \leq n_0} (\alpha m + \rho - \nfloor{\alpha m+\rho}).
\]

We take $w'$ to be the length $n_0$ prefix of $\omega$. To find the corresponding $\ell'$ we repeat the first part of the proof of Lemma \ref{le:Anml}. When we decrease $\rho$ as there, the first integer point the line $y = \alpha x + \rho$ reaches in the $n_0 \times n_0$-grid is $(n_0, j)$ for some integer $j$. Then, when we rotate the line counter-clockwise around $(n_0, j)$ to reach a line $\ell'$ with at least two grid points, all the grid points which the line meets are met from the above. Hence the claim follows. 
\end{proof}

\section{The enumeration of finite Sturmian words} \label{finite}
Now we are ready to give a new proof of the following enumeration formula for finite Sturmian words.
\begin{thm}
\label{th:Stnfact}
For all $n \geq 1$, we have
\begin{equation}\label{210620111029}
|A_n| = 1 + \sum_{k = 1}^n \sum_{l = 1}^k \varphi(l) = 1 + \sum_{k=1}^{n}(n+1 - k)\varphi(k).
\end{equation}
\end{thm}

This counting formula is proved by calculating the sum on the right hand side of \eqref{eq:Ansumml}. To do this, we use the next two lemmas.

\begin{lem}\label{210620111027}
Let $\ell \in \calL_n$ be the line $y = \alpha x + \rho$, in which case $\alpha \in [0,1]$ and $\rho \in [0,1]$. Then we have
\[
\nabs{m(\ell)} = 
\begin{cases}
0 & \text{if $\alpha = \rho = 1$} \\
1 & \text{if $\alpha = \rho = 0$} \\
\Z_n(\ell)-1 & \text{if $\alpha, \rho \in (0, 1)$} \\
\Z_{n}(\ell)-\Z_{n/2}(\ell) & \text{if $\rho = 0$ and $\alpha \in (0, 1]$} \\
\Z_{n/2}(\ell) - 1 & \text{if $\rho = 1$ and $\alpha \in [0, 1)$.} \\
\end{cases} 
\]
\end{lem}
Notice that the remaining options for $\alpha$ and $\rho$ above, namely that $\alpha = 0$ and
$\rho \in (0,1)$ or that $\alpha = 1$ and $\rho \in (0,1)$, do not correspond to any line in $\calL_{n}$.
\begin{proof}
If $\alpha=\rho = 1$, then $\ell$ is the line $y = x + 1$, which does not have any image words, i.e.,  $\nabs{m(\ell)}=0$. 
If $\alpha = \rho = 0$, then $\ell$ maps only to the word~$\word{0}^{n}$, so that $\nabs{m(\ell)}=1$.  
If $\alpha, \rho \in (0,1)$, then a line defining $w\in m(\ell)$ has two possibilities:
\begin{itemize}
\item It can first go just below one to $\Z_{n/2}(\ell)-1$ first integer
  points of $\ell$ and then go just above the rest. There are $\Z_{n/2}(\ell)-1$ such words.
\item It can first go just above $\Z_{n/2}(\ell)+1$ to $\Z_n(\ell)$
  integer points of $\ell$ and then go just below the rest. There are $\Z_{n}(\ell)-\Z_{n/2}(\ell)$ such words.
\end{itemize}
Thus there are altogether $\Z_n(\ell)-1$ words in $m(\ell)$ when $\alpha, \rho \in (0,1)$.

If $\rho = 0$ and $\alpha \in (0, 1]$, then $\ell$ is given by $y=\alpha x$. Then $\ell$ maps only to words whose
defining line starts by going above some integer
points of $\ell$ since the defining line must go above origin, which is
the first integer point of the line. Therefore $\ell$ maps to $\Z_{n}(\ell)-\Z_{n/2}(\ell)$ words.

If $\rho = 1$ and $\alpha \in [0,1)$, then $\ell$ is the line $y = \alpha x + 1$, which maps only to words whose
defining line first goes below some points in the line. There are $\Z_{n/2}(\ell) - 1$ such words.

\end{proof}

\begin{lem}
\label{le:countPhi}
For all $n \geq 1$, we have
\begin{equation}
\label{eq:countPhi}
\sum_{\substack{\ell \in \calL_n \\ \alpha \in
    (0, 1], \rho \in [0, 1)}} (\Z_n(\ell)-1) = \sum_{i = 1}^n \sum_{j = 1}^i \varphi(j).
\end{equation}
\end{lem}
\begin{proof}
For $i \in \{1, \dotsc, n\}$, let $\calL_i^\ast$ be the set of all those lines in $\calL_i$ with $\alpha \in (0, 1]$ and $\rho \in [0, 1)$ that have an integer point $(i, j)$ for some $j \in \{1, \dotsc, i\}$. A line $\ell$ counted on the left hand side of \eqref{eq:countPhi} belongs to $\Z_n(\ell)-1$ sets in the collection $\{\calL_i^\ast \colon 1 \leq i \leq n\}$ (indeed $\ell \in \calL_i^\ast$ if and only if $i$ is an $x$-coordinate of a non-leftmost grid point of $\ell$). Hence
\begin{equation}
\label{eq:ZnLast}
\sum_{\substack{\ell \in \calL_n \\ \alpha \in
    (0, 1], \rho \in [0, 1)}} (\Z_n(\ell)-1) = \sum_{i = 1}^n \nabs{\calL_i^\ast}.
\end{equation}

The set $\calL_i^\ast$ consists of lines $y = \frac{b}{a} x + \frac{c}{a}$ with $a, b, c$ satisfying
\[
0 < b \leq a \leq i, \quad 0 \leq c < a, \quad \gcd(a, b)=1, \quad \text{and} \quad c \equiv -bi \pmod{a},
\]
where the condition $a \leq i$ is needed since the line must contain an integer point with $x$-coordinate $\in \{0, 1, \dotsc, i-1\}$ besides the one with $x$-coordinate $i$. Now $a$ and $b$ determine $c$ uniquely, so that
\[
|\calL_i^\ast| = \sum_{a = 1}^i \sum_{\substack{1 \leq b \leq a \\ \gcd(b, a) =1}} 1 = \sum_{a=1}^i \varphi(a),
\]
and the lemma follows by combining this with \eqref{eq:ZnLast}. 
\end{proof}



We are finally ready to wrap up the proof of Theorem~\ref{th:Stnfact}.

\begin{proof}[Theorem~\ref{th:Stnfact}]
First notice that, for $\alpha \neq 0, 1$, the cardinalities of the sets $m(\ell:y = \alpha x)$ and $m(\ell:y = \alpha x + 1)$ add up to $\Z_n(\ell:y = \alpha x) - 1$.
Likewise, the cardinalities of $m(\ell:y = x)$ and $m(\ell:y =  1)$ add up to $\Z_n(\ell:y = x) - 1$.
Therefore, using Lemma~\ref{210620111027}, we get
\begin{align*}
\sum_{\ell \in \calL_n} \nabs{m(\ell)} &=
\nabs{m(\ell \colon y = 0)} + 
\sum_{\substack{\ell \in \calL_n \\ \alpha, \rho  \in (0, 1) }} \nabs{m(\ell)}
+ \sum_{\substack{\ell \in \calL_n \\ \rho = 0 , \alpha \in (0, 1]}} \nabs{m(\ell)}
+ \sum_{\substack{\ell \in \calL_n \\ \rho = 1 , \alpha \in [0, 1)}} \nabs{m(\ell)}
 \\
&=
1 +
\sum_{\substack{\ell \in \calL_n \\ \alpha, \rho  \in (0, 1) }} \nabs{m(\ell)}
+ \sum_{\substack{\ell \in \calL_n \\ \rho = 0 , \alpha \in (0, 1)}} \nabs{m(\ell)} + \nabs{m(\ell:y=x)} \\
&\phantom{= 1 }+ \sum_{\substack{\ell \in \calL_n \\ \rho = 1 , \alpha \in (0, 1)}} \nabs{m(\ell)} + \nabs{m(\ell:y = 1)} \\
&=
1 +
\sum_{\substack{\ell \in \calL_n \\ \alpha, \rho  \in (0, 1) }} \nabs{m(\ell)}
+ \sum_{\substack{\ell \in \calL_n \\ \rho = 0 , \alpha \in (0, 1)}} \bigl(\Z_{n}(\ell)-1\bigr) +  \bigl(\Z_n(\ell:y = x) - 1\bigr)\\
&=
1 + \sum_{\substack{\ell \in \calL_n \\ \alpha \in (0, 1], \rho \in [0, 1)}} \bigl(\Z_n(\ell) - 1\bigr).
\end{align*}
This identity together with Lemmas~\ref{le:Anml} and~\ref{le:countPhi} readily implies~\eqref{210620111029}, completing the proof. 
\end{proof}

\section{The enumeration of palindromic Sturmian words} \label{palindrome}
In the first part of this section we study how palindromes behave with respect to the mapping $m$. First we show that any palindrome must go through all integer points of the corresponding line $\ell \in \calL_n$. 
\begin{lem}\label{280820111347}
Let $w$ be a Sturmian palindrome of length $n$, and let $\ell \in \calL_{n}$ be such that  $w \in m(\ell)$.  
Then the broken line of $w$ goes through all integer points of~$\ell$. In particular, there are no palindromes in the set $m(\ell: y = \alpha x + 1)$ for any $\alpha \in [0,1]$.
\end{lem}
\begin{proof}
Let $\ell' \in \calS_{n}$ be a line defining $w$. If, contrary to what we want to prove, $w$ does not go through all integer points of $\ell$, then there  exist  three grid points $p_{1}, p_{2}, p_{3} \in \ell$ such that $i_{1} \leq n/2 < i_{2}$ and the line $\ell'$ goes just above $p_{1}$ and $p_{2}$ and below $p_{3}$. Let us assume that  $p_{3}$ is the first grid point in~$\ell$ with this property.
Also, write $\ell = \frac{b}{a}x + \frac{c}{a}$ with $(a,b) = 1$.

Suppose first that $i_{3} > i_{2}$. Then the prefix of $w$ of length $i_{3} - 1$ has period~$a$, but the period breaks at position $i_{3}$. 
This means that, while every block of length $a$ in the first half of $w$ has exactly $b$ occurrences of the letter $\word{1}$, 
there is a block of length~$a$ in the second half of $w$ with only $b-1$ occurrences of~$\word{1}$. Thus $w$ is not a palindrome, a contradiction.

The other possibility for $i_{3}$ is that $i_{3} < i_{1}$. This situation is analogous to the first case, and hence contradictory; we omit details.  
\end{proof}

The previous lemma in particular implies that every line $\ell \in \calL_n$ corresponds to at most one palindrome. Next we figure out which lines actually give a palindrome.
\begin{lem}\label{310820111103}
 There is a bijective correspondence between the palindromic Sturmian words $\neq \word{1}^{n}$ of length $n\geq 1$ and 
 the lines $\ell \in \calL_{n}$ given by $y = \frac{b}{a}x + \frac{c}{a}$ satisfying $(a,b) = 1$,  $0\leq b < a \leq n$, $0\leq c < a$, and
\begin{equation}\label{230820111329}
  2 c \equiv -bn -1 \pmod{a}.
\end{equation}
\end{lem}

\begin{proof}
 Let $w \neq \word{1}^{n}$ be a Sturmian palindrome of length $n$. Then $w \in m(\ell)$ for a line $\ell \in \calL_{n}$ with equation $y = \frac{b}{a}x + \frac{c}{a} \in \calL_{n}$, where $(a, b) = 1$ and
 $0 \leq b < a \leq n$. Furthermore, Lemma~\ref{280820111347} implies that the broken line of $w$ goes through all grid points of $\ell$ and $0 \leq c < a$.
In what follows, we will verify~\eqref{230820111329}. 

Since $b x + c \pmod{a}$ ranges all possible values $0,1,\ldots, a-1$  when $x$ ranges $0, 1, \ldots, n$, there exists an integer $k$ such that the vertical distance between $\ell$ and
the broken line of $w$ is $(a-1)/a$ at $x=k$, see Figure~\ref{051120101420}.

Rotate the plane by $180\,^{\circ}$ and, if necessary, shift the grid so that the obtained broken line starts from the origin, see Figure~\ref{051120101421}. 
The new broken line corresponds to the word $w^{R}$, the reversal of $w$.
Let $\ell^{R}$ denote the line obtained from $\ell$. Furthermore, let  $\ell' \in \calL_n$ be the line for which $w^{R} \in m(\ell')$, see Figure~\ref{16082011}.

 \begin{figure}[h]
 \centering
 \begin{tikzpicture}[scale=1.5]
 \draw[step=.5cm,gray,very thin] (-1,-0.5) grid (1.5,1.5); 
 \draw (-1, -0.33) -- (1.5, 1.33);
 \draw[semithick] (-1, -0.5) -- (-0.5, 0) -- (0, 0) -- (0.5, 0.5) -- (1, 1) -- (1.5, 1);
 \filldraw (-0.5, 0) circle(1pt);
 \filldraw (1, 1) circle(1pt);
 \draw[semithick] (0, 0) -- (0, 0.33); 
 \draw (1.6, 0.165) node[anchor=west]{$(a-1)/a$};
 \draw[->] (1.6, 0.165) -- (0.02, 0.165);
  \draw (1.5, 1.4) node[anchor=west]{$\ell$};
 \draw (1.5, 1) node[anchor=west]{$w$};
 \end{tikzpicture}
 \caption{Before rotation. Here $\ell = \frac{2}{3}x + \frac{1}{3}$ and the broken line corresponds to $w=\word{10110}$.}\label{051120101420}
 \end{figure}
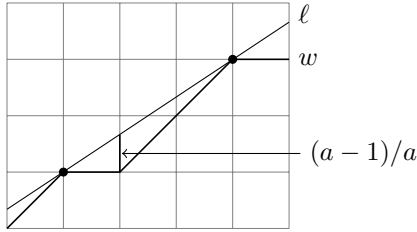

\begin{figure}[h]
 \begin{tikzpicture} [scale=1.5]
 \pgftransformrotate{180}
 \draw[step=.5cm,gray,very thin] (-1,-1) grid (1.5,1); 
 \draw (-1, -0.33) -- (1.5, 1.33);
 \draw[semithick] (-1, -0.5) -- (-0.5, 0) -- (0, 0) -- (0.5, 0.5) -- (1, 1) -- (1.5, 1);
 \filldraw (-0.5, 0) circle(1pt);
 \filldraw (1, 1) circle(1pt);
  \draw[semithick] (0, 0) -- (0, 0.33);
 \draw (-1, -0.2) node[anchor=west]{$\ell^{R}$};
 \draw (-1, -0.6) node[anchor=west]{$w^{R}$};
 \end{tikzpicture}
 \caption{After rotation. Here the broken line corresponds to $w^{R} = \word{01101}$, and $\ell^{R}$ denotes the line $\ell$ after rotation.}\label{051120101421}
\end{figure}
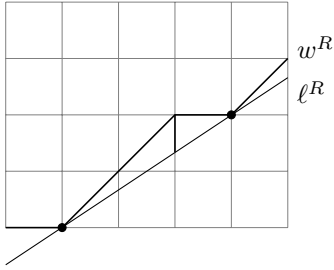
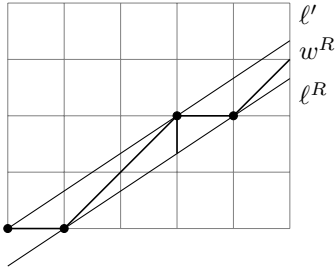
\begin{figure}
 \begin{tikzpicture} [scale=1.5]
 \draw[step=.5cm,gray,very thin] (0,0.5) grid (2.5,2.5); 
 \draw (0, 0.5) -- (2.5, 2.166);
 \draw (0, 0.167) -- (2.5, 1.83);
 \draw[semithick] (0, 0.5) -- (0.5, 0.5) -- (1, 1) -- (1.5, 1.5) -- (2, 1.5) -- (2.5, 2);
 \filldraw (0, 0.5) circle(1pt);
 \filldraw (1.5, 1.5) circle(1pt);
 \filldraw (0.5, 0.5) circle(1pt);
 \filldraw (2, 1.5) circle(1pt);
  \draw[semithick] (1.5, 1.166) -- (1.5,1.5);
 \draw (2.5, 1.7) node[anchor=west]{$\ell^{R}$};
 \draw (2.5, 2.1) node[anchor=west]{$w^{R}$};
 \draw (2.5, 2.4) node[anchor=west]{$\ell'$};
 \end{tikzpicture}
 \caption{The line $\ell' = \frac{2}{3}x$ added to the grid in Figure~\ref{051120101421}.}\label{16082011}
\end{figure}

By Lemma \ref{le:Anml}, the word $w$ is a palindrome iff $w = w^{R}$ iff $\ell' = \ell$ and $w^{R}$ goes through all integer points of $\ell'$.
In particular $\ell'$ must have the same slope $b/a$ as~$\ell$. The original $\ell$ goes through the point $(n, bn/a+c/a)$, hence $\ell^{R}$ starts from the point $\bigl(0, -(bn/a+c/a) \pmod{1}\bigr)$. 
As~$\ell'$ must touch but not intersect the broken line of $w^{R}$, it must be $(a-1)/a$ above $\ell^{R}$. Hence the equation of
$\ell'$ is 
\[
y \equiv \frac{b}{a} x - \left(\frac{b}{a} n + \frac{c}{a}\right) +
    \frac{a-1}{a} \pmod{1} . 
\]
This should be the same as the original line so that 
\[
- \left(\frac{b}{a} n + \frac{c}{a}\right) + \frac{a-1}{a} \equiv
\frac{c}{a} \pmod{1} \iff 2 c \equiv -bn -1 \pmod{a},
\]
and so \eqref{230820111329} is verified.

We have thus established a mapping between Sturmian palindromes $\neq \word{1}^{n}$ of length $n$ and the lines satisfying conditions described above. Injectivity of the mapping follows from Lemma~\ref{le:Anml} and surjectivity is shown by tracing back the steps taken above. 
\end{proof}

Now we are ready to  give a new proof of the following theorem by de Luca and De Luca~\cite{deLDeL2006}.

\begin{thm}\label{28082011}
The number of Sturmian palindromes of length $n$ is
\begin{equation}\label{05092011}
1 + \sum_{k=0}^{\nceil{n/2}-1}\varphi(n-2k).
\end{equation}
\end{thm}

\begin{proof}
Let us denote the set of lines in $\calL_{n}$ satisfying the conditions of Lemma~\ref{310820111103} by $\calP_{n}$. We will prove the claim by enumerating this set.

For every $i$, we count the number of the lines $\ell \in \calP_n$ whose leftmost
integer point is $(i, j)$ for some $0\leq j \leq i$. Thus $i$ runs through values 0,1,\ldots, $\lceil n/2 \rceil - 1$. Line through $(i, j)$ with slope $b/a$ (where $\gcd(a, b) = 1$) has equation
\[
y - j = \frac{b}{a}(x-i) \iff y = \frac{b}{a}x + \frac{aj-bi}{a},
\]
so that~\eqref{230820111329} can be expressed as
\[
-2bi \equiv -bn - 1 \pmod{a} \iff b(n-2i) \equiv -1 \pmod{a}. 
\]
This is soluble on $b$ if and only if $\gcd(a, n-2i) = 1$ in which case it
has a unique solution satisfying $0\leq b < a$ and $\gcd(b, a) = 1$. Notice also that $j$ is determined uniquely by the condition $0 \leq aj-bi < a$.

Since $\{i+ak \ \vert \ k \in \mathbb{Z}\} \cap [0, n]$ is the set of $x$-coordinates of all integer points of $\ell$ in the $n \times n$-grid, the point $(i, j)$ is the leftmost of at least two integer points if and only if
\[
i \leq a-1 \text{ and } i \leq n-a \iff i+1 \leq a \leq n-i.
\] 
Therefore, the number of lines in $\calP_{n}$ is
\[
\sum_{i = 0}^{\lceil n/2 \rceil-1} \sum_{\substack{a = i+1 \\ \gcd(a, n-2i) =1}}^{n-i} 1 =
\sum_{i=0}^{\lceil n/2 \rceil-1} \varphi(n-2i).
\]
Recalling that $\calP_{n}$ does not include  the line in $\calL_{n}$ for the word $\word{1}^{n}$, we obtain the attested formula~\eqref{05092011}.
\end{proof}

\section{Return words}\label{return}
Recall the definition of a return word from the end of Section \ref{se:intro}. To get hold of occurrences of a factor $u$ in a given word, we prove the following lemma.
\begin{lem}
\label{le:startpoint}
Let $0 \leq b < a$ with $\gcd(b,a) = 1$ and let $u$ be a finite word defined by the line
\begin{equation}
\label{eq:faceq}
y = \frac{b}{a} x + \frac{c'}{a} \quad \text{with $0 \leq x \leq |u|$}
\end{equation}
for some $c'$. Then there are integers $0 \leq c_1 \leq c_2 < a$ such that \eqref{eq:faceq} represents $u$ if and only if $c' \in [c_1, c_2] \pmod{a}$. Furthermore, of the lines representing $u$, only the one with $c' = c_2$ goes through a point $(x, y)$ with $x \in \{0, \dotsc, |u|\}$ and $y \equiv (a-1)/a \pmod{1}$.
\end{lem}

\begin{proof}
Let $c_1$ be the smallest and $c_2$ the greatest $c' \in \{0, \dotsc, a-1\}$ such that \eqref{eq:faceq} represents $u$. Then there cannot be integer points between the lines \eqref{eq:faceq} with $c'=c_1$ and $c'=c_2$ since otherwise these two lines would represent different words. Hence \eqref{eq:faceq} represents $u$ also for any $c' \in [c_1, c_2]$.

For any $c$, the lines \eqref{eq:faceq} with $c' = c$ and $c' = c+1$ represent different words if and only if there is an integer point in the line with $c' = c+1$. This is equivalent to the line with $c' = c$ having a point $(x, y)$ with $x \in \{0, \dotsc, |u|\}$ and $y = (a-1)/a \pmod{1}$. Now the second assertion of the lemma follows from the first. 
\end{proof}

Let us now state and prove our final application, a new proof of the following theorem from \cite{JusVui2000, Vuillon2001}.
\begin{thm}
\label{th:rets}
Every Sturmian word has exactly two returns.
\end{thm}

\begin{proof}
Let $\omega$ be a Sturmian word and let $u$ be a factor of $\omega$. If $u$ had only one return, then $\omega$ would be eventually periodic, which is a contradiction. Therefore it is enough to prove that $u$ has at most two returns.

Clearly it is enough to consider a finite Sturmian factor $w$, and by Lemma \ref{le:redallints} we can assume that $w$ goes through all integer points of the line $\ell$ for which $w \in m(\ell)$. The case $w = \word{1}^{n}$ is trivial, so assume that $w \neq \word{1}^{n}$. Now the line $\ell$ is defined by an equation
\[
y = \frac{b}{a} x + \frac{c}{a}, \quad \text{where $0 \leq b < a, 0 \leq c < a$
  and $\gcd(b,a) = 1$.}
\]
We say that a factor beginning from position $i$ \emph{starts from} $(bi + c)/a \pmod{1}$. So, for example, $w$ starts from $c/a$.

By Lemma \ref{le:startpoint}, there are integers $0 \leq c_1 \leq c_2 < a$ such that a factor
of length $|u|$ is $u$ if and only if it starts from a point in $[c_1/a, c_2/a] \pmod{1}$. For $i = c_1, c_1+1, \dotsc, c_2$, let $v_i$ be the return word corresponding to the occurrence of $u$ starting from $i/a$.

If we move the line $\ell$ upwards by $1/a$ and consider also the word $w_+$ corresponding to the moved line, the only difference to the word $w$ is at points where
\[
y \equiv \frac{a-1}{a} \pmod{1} \quad \text{moves to} \quad  y \equiv 1 \pmod{1},
\]
i.e., the line reaches an integer point (and there $01 \to 10$). By Lemma \ref{le:startpoint} this must happen in $u$ starting from
$c_2/a$ and only in that $u$.

Comparing $w$ and $w_+$, those $u$ starting from $i/a$ with $i < c_2$ change to $u$
starting from $(i+1)/a$, so that $v_{i} \to v_{i+1}$. By above $v_{i+1}$ equals $v_i$
unless $v_i$ starts or ends with $c_2/a$. Let this happen for $v_{c_2}$ and $v_k$. We get
\[
v_{c_1} = v_{c_1+1} = \dotsb = v_{k} \neq v_{k+1} = v_{k+2} = \dotsb =  v_{c_2},
\]
so that there can be at most two different $v_i$. 
\end{proof}

To illustrate the proof of Theorem~\ref{th:rets}, let us take $w= \word{10001001001000100}$ and $u=\word{100}$. 
The broken line of $w$ is given by 
\[
\ell : y = \frac{3}{10} x + \frac{7}{10},
\]
see Figure~\ref{070920111425}. The numbers above the line $\ell$ are numerators of the ``starts'' of factors of $w$.

 \begin{figure}[!h]
 \centering
 \begin{tikzpicture} [scale=0.5]
 \draw[step=1cm,gray,very thin] (0,0) grid (17, 7); 
\pgfmathparse{17*0.3 + 0.7}
\edef\y{\pgfmathresult} 
\draw (0,  0.7) -- (17, \y);
\draw[semithick] (0.0, 0.0) -- (1, 1) -- (4, 1) -- (5,2) -- (7, 2) -- (8,3) -- (10,3) -- (11, 4) -- (14,4) -- (15,5) -- (17,5);
 \foreach \x in {0,1,2,...,17}
 {	
 	\pgfmathtruncatemacro{\y}{mod(3*\x+7, 10)}
 	\draw (\x, 1.2) + (0, 0.3*\x) node {\small{\y}};
}
\filldraw (1, 1) circle (2pt);
\filldraw (11, 4) circle (2pt); 
 \draw[semithick] (0, -1.2) -- (0, -0.1);
 \draw (0.5,-0.5) node {\word{1}}; 
 \draw (1.5,-0.5) node {\word{0}};
 \draw (2.5,-0.5) node {\word{0}};
 \draw (3.5,-0.5) node {\word{0}};
 \draw (2 , -1.5) node {$v_{7}$};
 \draw[semithick] (4, -1.2) -- (4, -0.1);
 \draw (4.5,-0.5) node {\word{1}};
 \draw (5.5,-0.5) node {\word{0}};
 \draw (6.5,-0.5) node {\word{0}};
  \draw (5.5 , -1.5) node {$v_{9}$};
  \draw[semithick] (7, -1.2) -- (7, -0.1);
 \draw (7.5,-0.5) node {\word{1}};
 \draw (8.5,-0.5) node {\word{0}};
 \draw (9.5,-0.5) node {\word{0}};
  \draw (8.5 , -1.5) node {$v_{8}$};
  \draw[semithick] (10, -1.2) -- (10, -0.1);
 \draw (10.5,-0.5) node {\word{1}};
 \draw (11.5,-0.5) node {\word{0}}; 
 \draw (12.5,-0.5) node {\word{0}}; 
 \draw (13.5,-0.5) node {\word{0}};
  \draw (12 , -1.5) node {$v_{7}$};
  \draw[semithick] (14, -1.2) -- (14, -0.1);
 \draw (14.5,-0.5) node {\word{1}};
 \draw (15.5,-0.5) node {\word{0}};
 \draw (16.5,-0.5) node {\word{0}};
  \draw (15.5 , -1.5) node {$v_{9}$};

 \end{tikzpicture}
\caption{Line $\ell$ and the broken line of $w$.}\label{070920111425}
 \end{figure}
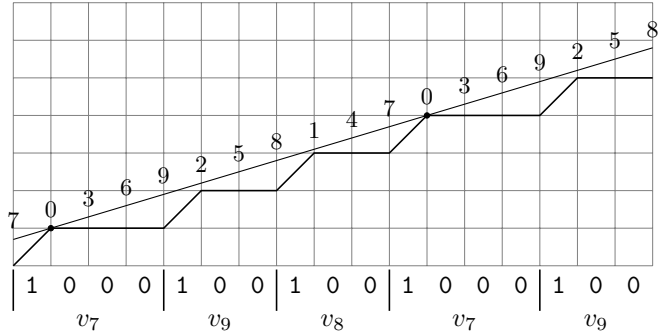

Now, the possible starting points for $u$ are $\frac{7}{10}$,  $\frac{8}{10}$, and $\frac{9}{10}$, so that $c_{1} = 7$ and $c_{2}=9$.
Next, by moving the line $\ell$ upwards by $\frac{1}{10}$, we reach the line given by $y = \frac{3}{10} x + \frac{8}{10}$; this is depicted in Figure~\ref{070920111429}.
In this transition, we see that $v_{7} \rightarrow v_{8}$ and $v_8 \rightarrow v_{9}$. Since $v_{7}$ ends with $\frac{c_{2}}{a} = \frac{9}{10}$, there is a change, so that $v_{7} \neq v_{8}$. The word $v_{8}$ does not start nor end end with $\frac{c_{2}}{a} = \frac{9}{10}$, so $v_{8} = v_{9}$.

 \begin{figure}[h!]
 \centering
 \begin{tikzpicture} [scale=0.5]
 \draw[step=1cm,gray,very thin] (0,0) grid (17, 7); 
\pgfmathparse{17*0.3 + 0.8}
\edef\y{\pgfmathresult} 
\draw (0,  0.8) -- (17, \y);
\draw[semithick] (0.0, 0.0) -- (1, 1) -- (3, 1) -- (4,2) -- (5,2) -- (7, 2) -- (8,3) -- (10,3) -- (11, 4) -- (13,4) -- (14,5) -- (15,5) -- (17,5);
 \foreach \x in {0,1,2,...,17}
 {	
 	\pgfmathtruncatemacro{\y}{mod(3*\x+8, 10)}
 	\draw (\x, 1.3) + (0, 0.3*\x) node {\small{\y}};
}
 \draw[semithick] (0, -1.2) -- (0, -0.1);
 \draw (0.5,-0.5) node {\word{1}}; 
 \draw (1.5,-0.5) node {\word{0}};
 \draw (2.5,-0.5) node {\word{0}};
  \draw (1.5 , -1.5) node {$v_{8}$};
  \draw[semithick] (3, -1.2) -- (3, -0.1);
 \draw (3.5,-0.5) node {\word{1}};
 \draw (4.5,-0.5) node {\word{0}};
 \draw (5.5,-0.5) node {\word{0}};
 \draw (6.5,-0.5) node {\word{0}};
  \draw (5 , -1.5) node {$v_{7}$};
  \draw[semithick] (7, -1.2) -- (7, -0.1);
 \draw (7.5,-0.5) node {\word{1}};
 \draw (8.5,-0.5) node {\word{0}};
 \draw (9.5,-0.5) node {\word{0}};
 \draw (8.5 , -1.5) node {$v_{9}$};
  \draw[semithick] (10, -1.2) -- (10, -0.1);
 \draw (10.5,-0.5) node {\word{1}};
 \draw (11.5,-0.5) node {\word{0}}; 
 \draw (12.5,-0.5) node {\word{0}};
 \draw (11.5 , -1.5) node {$v_{8}$};
 \draw[semithick] (13, -1.2) -- (13, -0.1); 
 \draw (13.5,-0.5) node {\word{1}};
 \draw (14.5,-0.5) node {\word{0}};
 \draw (15.5,-0.5) node {\word{0}};
 \draw (16.5,-0.5) node {\word{0}};
 \draw (15 , -1.5) node {$v_{7}$};

 \end{tikzpicture}
\caption{Moving the line upwards by $1/10$ reaching the line given by $y = \frac{3}{10} x + \frac{8}{10}$ and corresponding to the word $w_+ = \word{10010001001001000}$ .}\label{070920111429}
 \end{figure}
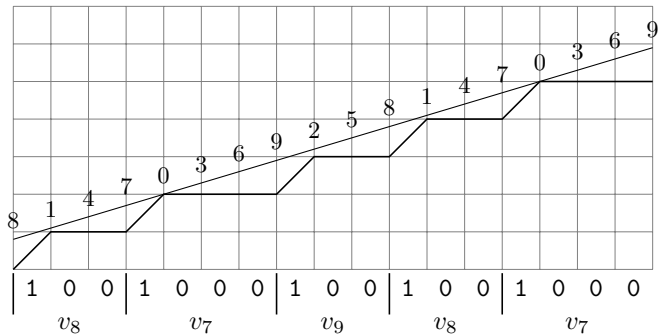

\section{Acknowledgements}
We thank the anonymous referees for helpful comments on the exposition of the paper and for pointing out the reference~\cite{Lipatov1982}.

\end{document}